\documentclass{amsart}

\usepackage{extarrows} 
\usepackage{amssymb}
\usepackage{amsfonts}
\usepackage{amsthm}
\usepackage{graphicx}
\usepackage{subcaption}
\usepackage{color}
\usepackage{url}
\usepackage{xspace}
\usepackage{microtype}
\usepackage{tikz}
\usepackage{enumitem}
\usepackage{hyperref}
\usepackage{pgf}
\usepackage{tikz-3dplot}
\usepackage{tikz-cd}
\usepackage{float}
\usepackage[utf8]{inputenc}

\definecolor{darkblue}{RGB}{0,0,160}
\hypersetup{
colorlinks,%
citecolor=black,%
filecolor=black,%
linkcolor=darkblue,%
urlcolor=darkblue
}

\newtheorem{thm}{Theorem}
\newtheorem*{thm*}{Theorem}

\newtheorem{question}[thm]{Question}
\newtheorem{conj}[thm]{Conjecture}

\theoremstyle{definition}
\newtheorem{example}[thm]{Example}
\newtheorem{remark}[thm]{Remark}
\newtheorem{defn}[thm]{Definition}

\newcommand{\ring}[1]{\ensuremath{\mathbb{#1}}}
\newcommand\RR{\ring{R}}
\newcommand\ZZ{\ring{Z}}

\newcommand\cL{{\mathcal L}}
\newcommand\cM{{\mathcal M}}

\newcommand\cP{{\mathcal P}}

\newcommand\cX{{\mathcal X}}

\newcommand\ba{{\bf a}}

\newcommand\bk{{\bf k}}

\newcommand\bx{{\bf x}}

\newcommand\fN{{\mathfrak N}}

\newcommand\fZ{{\mathfrak Z}}

\newcommand\cPc{{\mathcal P_{\textbf{c}}}}
\newcommand\Lcons{{\cL_{\text{cons}}}}

\newcommand\lin{{\text{lin}}}

\newcommand{\simeqd}{\mathrel{\rotatebox[origin=c]{-90}{$\simeq$}}}

\begin{document}

\title[Duality in mass-action networks]{Duality in mass-action networks}
\author{Alexandru Iosif}
\address{Area of Applied Mathematics, Universidad Rey Juan Carlos, Madrid, Spain}
\email{alexandru.iosif@urjc.es}


\begin{abstract}
Mass-action networks are special cases of chemical reaction networks. For 
these systems, we argue that conserved quantities are dual to internal 
cycles. We introduce maximal invariant polyhedral 
supports, and we conjecture that there is a duality relation between 
preclusters and maximal invariant polyhedral supports. Given the 
close relation between maximal invariant polyhedral supports and siphons, we 
also conjecture that siphons and preclusters are dual objects.
\end{abstract}



\maketitle

\section{Introduction}

The syntagma ``mass-action'' originates from the work of Guldberg and
Waage back in the XIX century, which culminated with the \emph{Law of
  mass-action}. This law is a rough dynamical approximation to the way
molecules interact and states that ``The rate at which a unit of a
chemical species is consumed or produced by a chemical reaction is
proportional to the product of the concentrations of the reactants''
\cite{erdi1989mathematical}. Interestingly enough, this principle has
been successfully applied to biological problems (see
\cite{mclean1938application} for an early review from 1938 and
\cite{voit2015150} for a more recent review of these applications; the
close relation between mass-action networks and algebraic geometry has
been reviewed in \cite{dickenstein2016biochemical}).

Mathematically, a \emph{chemical reaction} is a transformation between two 
sets of chemical species; a set of chemical reactions is called \emph{chemical
reaction network}. Often, chemical reaction networks are regarded as
dynamical systems in which the concentrations of the chemical species
are described by partial differential equations depending on the position
and time. When the chemical system is spatially homogeneous, one can
instead use ordinary differential equations. This rough approximation
is often made in Biochemistry where the rate at which a concentration
$x_i$ changes is
\[
\dot x_i = \sum_j \phi_{ij} (x_1,\ldots,x_n),
\]
where $j$ runs through all chemical reactions involving $x_i$ and
$\phi_{ij}$ is the rate at which $x_i$ is consumed or produced by the 
$j^{\text{th}}$ reaction. In the special case of the mass-action Law, the 
functions $\phi$ are monomials with real coefficients.

In this paper, we study a duality theory for mass-action networks (that is, 
chemical reaction networks whose dynamics is of mass-action type). We review 
some of the combinatorial and algebraic work in the literature, and show 
that two of the combinatorial objects studied in the literature are dual. More 
precisely, we show that conserved quantities are dual to minimal 
cycles. For conservative mass-action networks that do not have two species 
with exactly the same rates and have at least one positive steady state, 
we also conjecture that there is a duality 
relation between preclusters and maximal invariant 
polyhedral supports. Given the close relation between maximal invariant 
polyhedral supports and siphons, we finally conjecture that there is a duality 
relation between siphons and preclusters.

\section{Combinatorial view of mass-action networks}

Consider a set of $n$ chemical species,
$\{X_1,X_2,\ldots, X_n\}$. The set $\mathcal C$ of
chemical complexes on $\{X_1,X_2,\ldots, X_n\}$ is the free commutative monoid
generated by $\{X_1,X_2,\ldots, X_n\}$ whose operation is denoted by
$+$. Thus, for example, $X_1+X_1+X_3=:2X_1+X_3$ denotes a chemical complex.

\begin{defn}
\label{defn:crn}
A \emph{chemical reaction network} is any relation
\[
\fN\subseteq \mathcal C^2.
\]
We usually denote by $r$ the cardinal of $\fN$. In the 
following, $\fN$ is going to be finite.
\end{defn}

\begin{example}
\label{ex:1}
The following chemical reaction network models one of the ways in
which a protein $A$ can be phosphorylated by an enzyme $E$.
\[
\begin{tikzcd}[row sep=small, column sep=small, every arrow/.append
style={shift left=.75}]
A+E \arrow{rr} && AE \arrow{ll} \arrow{rr} &&
A_p+E
\end{tikzcd}
\]
Here, $AE$ denotes an intermediate species and $A_p$ denotes the
phosphorylated protein. Hence, the set of species is
$\{A, E, AE, A_p\}$. Complexes will be $\ZZ_{\ge0}-$linear
combinations of these species. Thus, $A+E$, $AE$, and $A_p+E$ are
examples of complexes.
\end{example}

\begin{remark}
In order to keep notation easy, whenever considering a chemical
reaction network it is assumed an underlying ordering of both the
species and the reactions; otherwise, we would need to take quotients
by symmetric groups.
\end{remark}

Now, given a field $\mathbb K$, there is a natural map
\[
\begin{array}{cccc}
  \mu:&\mathcal C&\to&\mathbb K[x_1,\ldots,x_n]\\
      & \sum_{i=1}^n c_iX_i&\mapsto&\prod_{i=1}^nx_i^{c_i}.
\end{array}
\]
This map can be extended coordinate-wise to elements of $\fN$. We also use the 
following notation:
\begin{enumerate}
    \item $\cX := \{x_1,\ldots,x_n\}$,
    \item $\bx := (x_1,\ldots, x_n)$,
    \item and, given a monomial $m$, we denote its set of variables by $\cM$.
\end{enumerate}

\begin{defn}
\label{defn:mass-action-net}
The \emph{mass-action network} associated with the chemical reaction network $\fN$ is
\[
\mu(\mathfrak N)\times \mathbb R^r_{>0},
\]
where $r=|\fN|$. Since this provides us with a bijection between chemical 
reaction networks and mass-action networks, when clear from the context, we 
are going to abuse the notation and use $\fN$ to denote both $\fN$ and 
$\mu(\fN)\times\mathbb R^r$.
\end{defn}

\begin{remark}
The set $\RR_{>0}^r$ from Definition~\ref{defn:mass-action-net} is to
be interpreted as a space of parameters. The coordinates of this
space are called \emph{rate constants}, and they are usually denoted
by $k_1$, $\ldots$, $k_r$. In mathematical biology literature, arrows
are usually identified with $k_i$.
\end{remark}

\begin{example}
\label{ex:2}
Consider the chemical reaction network of Example
\ref{ex:1}.
We have
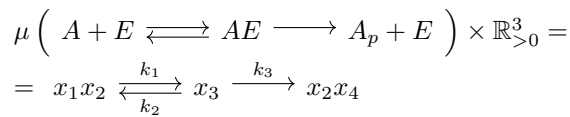
\begin{equation}
\label{ex:1site}
\begin{array}{l}
  \mu
\left(
  \begin{tikzcd}[row sep=small, column sep=small, every arrow/.append
  style={shift left=.75}]
  A+E \arrow{rr} && AE \arrow{ll} \arrow{rr} && A_p+E
  \end{tikzcd}
  \right)\times \RR^3_{>0}=
  \\
=\begin{tikzcd}[row sep=small, column sep=small, every arrow/.append
  style={shift left=.75}]
  x_1x_2 \arrow{rr}{k_1} && x_3 \arrow{ll}{k_2} \arrow{rr}{k_3} && x_2x_4
  \end{tikzcd}
\end{array}    
\end{equation}

\end{example}

Hence, a mass-action network is equivalent to a finite
directed graph whose vertices are labelled by monomials and whose
edges are labelled by \emph{rate constants}. By abusing the notation, $\fN$ 
 also denotes the corresponding graph.

Note that, to each chemical reaction network $\fN$, we can associate two
matrices of nonnegative integers as follows:
\begin{enumerate}[label=\roman*]
\item Consider all the sources of $\fN$ and represent the
$\ZZ_{\ge0}$ coefficients of the $i^{\text{th}}$ source as the column
of a matrix $Y_e\in\ZZ_{\ge0}^{n\times r}$. $Y_e$ is called the \emph{educt matrix}.
\item Then consider all the sinks of $\fN$ and represent the
$\ZZ_{\ge0}$ coefficients of the $i^{\text{th}}$ sink as the column of
a matrix $Y_p\in\ZZ_{\ge0}^{n\times r}$. $Y_p$ is called the \emph{product matrix}.
\end{enumerate}
The orderings of the elements of the previous matrices can be chosen so that they coincide with the orderings in $\fN$. Thus, every chemical reaction network can be expressed as
\[
\begin{tikzcd}
\prod_{i=1}^n x_i^{(Y_e)_{i_1}}\arrow{r}{k_1}&\prod_{i=1}^n x_i^{(Y_p)_{i_1}}
\end{tikzcd}, \ \ldots,
\begin{tikzcd}
\prod_{i=1}^n x_i^{(Y_e)_{i_r}}\arrow{r}{k_r}&\prod_{i=1}^n x_i^{(Y_p)_{i_r}},
\end{tikzcd}
\]
or, in short, as
$\begin{tikzcd}
x^{Y_e}\arrow{r}{k}& x^{Y_p}.
\end{tikzcd}$
If we use the notation
$\bk = (k_1,\ldots,k_r)$, any mass-action network can be uniquely
expressed as a quadruple $(Y_e,Y_p,\bk,\bx)$. The difference $Y_p-Y_e$ is
known as the \emph{stoichiometric matrix} of $\fN$. Hence, we will
call the quadruple $(Y_e,Y_p,\bk,\bx)$ \emph{the stoichiometric
representation of $\fN$}.

\begin{remark}
Note that, to define a chemical reaction network, in principle, one
does not need the level of generality introduced in this paper. However, we often argue by taking sums of reactions
(e.g., the alternative proof from the author's thesis of the well known fact that internal cycles
are non negative elementary flux modes~\cite{schuster1994elementary}
makes use of such sums). It is our philosophy that the language from
Definition~\ref{defn:crn} clarifies the underlying combinatorial
structure of the proofs. Moreover, this formalism could potentially provide a way of representing polynomial ideals by such directed graphs, translating algebraic properties to combinatorial ones and vice versa. This translation is implicit in chemical reaction network papers but, to our knowledge, it has not been used so far in a pure mathematical framework.
\end{remark}

\begin{example}
The educt and product matrices of network~\eqref{ex:1site} are
\[
Y_e=
\left(
  \begin{array}{ccc}
    1&0&0\\
    1&0&0\\
    0&1&1\\
    0&0&0
  \end{array}
\right)
\text{ and } \
Y_p=
\left(
  \begin{array}{ccc}
    0&1&0\\
    0&1&1\\
    1&0&0\\
    0&0&1
  \end{array}
\right).
\]
\end{example}

The reason we have introduced this notation is that the dynamics
of a mass-action network is given by the ODE system
\begin{equation}
\dot \bx^T = (Y_p-Y_e)\text{diag}(\bk)\left(\bx^T\right)^{Y_e},
\label{eq:dyn-sys}
\end{equation}
where $\left(\bx^T\right)^{Y_e}$ is, by definition, a column vector of length 
$r$ such that
\[
\left(\left(\bx^T\right)^{Y_e}\right)_j = \prod_{i=1}^n x_i^{(Y_e)_{ij}}.
\]

\begin{example}
The dynamics of network~\eqref{ex:1site} is given by
\[
\begin{array}{ll}
  \dot x_1=&-k_1x_1x_2+k_2x_3\\
  \dot x_2=&-k_1x_1x_2+k_2x_3+k_3x_3\\
  \dot x_3=&k_1x_1x_2-k_2k_3-k_3x_3\\
  \dot x_4=&k_3x_3.
\end{array}
\]
\end{example}

Realistic models tend to be large; hence, often, one does not intend to
solve them, but rather studies their asymptotic behaviour, like their steady
states. Steady states are defined as the nonnegative solutions of the following parametric family of systems
\begin{equation}
(Y_p-Y_e)\text{diag}(\bk)\left(\bx^T\right)^{Y_e} = 0.\label{eq:st-st}
\end{equation}

Note that, when one replaces the dynamical system \eqref{eq:dyn-sys}
with the semialgebraic static system~\eqref{eq:st-st}, information about the
conserved quantities is lost. Therefore, one has to regard the
solutions of system~\eqref{eq:st-st} with respect to such conserved
quantities. In biochemical reaction networks one usually restricts
conservation laws to linear ones, that is, those which are of the form
$z_1x_1 + \ldots z_nx_n = c$, where $z$ is an element of the
conservation space $\Lcons:=\text{left--ker} (Y_p-Y_e)$ and $c$ is the value
of the conserved quantity, that is, given the initial conditions of
the corresponding ODE system, each one of such $c$ values is a
constant along trajectories (see \cite{desoeuvres2024computational, desoeuvres2024reduction} for two complementary studies of nonlinear conservation laws). Here, given a matrix $A$, by
$\text{left--ker}(A)$  we understand the linear space formed by vectors $\ba$ 
such that $A^T\ba^T=0$. As $\bx$ are nonnegative, it only makes
sense to consider those solutions to the conservation equations
contained in $\RR^n_{\ge0}$. The sets of the solutions of these linear
conservation laws (that is, systems of inhomogeneous linear
equations) inside $\RR^n_{\ge0}$ give rise to polyhedra called
\emph{invariant polyhedra} (note that we intersect a linear space with
the cone $\RR^n_{\ge0}$). As the values $c$ are varied, the
associated invariant polyhedra varies as well, hence, as the values of the
conserved quantities change, the shape of the invariant polyhedron
also changes. However, it is known that there is only a finite number
of combinatorial changes of such polyhedra \cite{de2009graphs}.

Sometimes, in the literature, $\cP_{{\bf{x}}(0)}$ denotes the 
invariant polyhedron corresponding to the vector ${\bf{x}}(0)$ of initial 
conditions (this is the case, for instance, in \cite{alg-030}). Now, if 
${\bf{x}}(0)$ and ${\bf{y}}(0)$ are two vectors of initial conditions such 
that $\bf z \cdot {\bf{x}}(0)=\bf z \cdot {\bf{y}}(0)$ for all $\bf z$ in the 
left kernel of $Y_p-Y_e$, then $\cP_{{\bf{x}}(0)}=\cP_{{\bf{y}}(0)}$. If we 
choose a generating system $\fZ=({\bf z}_1,\ldots,{\bf z}_q)$ for the left 
kernel of $Y_p-Y_e$ and we denote
\[
{\bf c} := \left({\bf z}_1\cdot {\bf{x}}(0),\ldots,{\bf z}_q\cdot {\bf{x}}
(0)\right),
\]
then we can define the following equivalence class
\begin{multline*}
    \cPc|_{\fZ}:=\left[\cP_{{\bf{x}}(0)}\right]=\\
    =\left\{
    \cP_{{\bf{y}}(0)}|\left({\bf z}_1\cdot 
    {\bf{x}}(0),\ldots,{\bf z}_q\cdot {\bf{x}}(0)
    \right)=
    \left(
    {\bf z}_1\cdot {\bf{y}}(0),\ldots,
    {\bf z}_q\cdot {\bf{y}}(0)\right)= {\bf c}
    \right\}.
\end{multline*}
When it is clear from the context the choice of $\fZ$, we write $\cPc$ instead 
of $\cPc|_{\fZ}$.

\begin{remark}
Often, instead of using any generating set of $\fZ$, it is more useful to 
choose the extreme rays of the \emph{conservation cone}, $\Lcons\cap\RR^{\dim \Lcons}_{\ge0}$.
\end{remark}

\begin{remark}
Note that when the trajectories of \eqref{eq:dyn-sys} are bounded
(i.e., when the system is \emph{conserved}), it is enough to restrict
$z$ to the conservation cone, and invariant polyhedra become
polytopes.    
\end{remark}

\begin{example}
  Consider the following mass-action network
  \[
\begin{tikzcd}[row sep=small, column sep=small]
  x_1x_4 && x_3 \arrow{ll}[swap]{k_1} \arrow{rr}{k_2} && x_2x^2_4
  \end{tikzcd}
\]
The associated ODE system is
\[
  \left\{
\begin{array}{l}
  \dot x_1=k_1x_3,\\
  \dot x_2=k_2x_3,\\
  \dot x_3=-k_1x_3-k_2x_3,\\
  \dot x_4=k_1x_3+2k_2x_3,
\end{array}
\right.
\qquad
x_i\ge0, \ i\in[4].
\]
It has two linearly independent conservation laws:
\[
  \begin{array}{l}
    x_1+x_2+x_3=c_1\\
    x_1+2x_2-x_4=c_2,
  \end{array}
\]
where $c_1$ and $c_2$ are constant along the trajectories.  Hence, for
each pair of values $\mathbf{c}:=(c_1,c_2)\in\RR^2_{\ge0}$, there is an
invariant polyhedron $\cP_{\mathbf c}$ inside which the trajectories of
the ODE system live, defined by the following system:
\[
  \left\{
  \begin{array}{l}
    x_1+x_2+x_3=c_1\\
    x_1+2x_2-x_4=c_2\\
    x_i\ge0, \text{ for } i\in\{1,2,3,4\}.
  \end{array}
  \right.
\]
It is clear that each $\cP_{\mathbf c}$ lives inside a two dimensional affine space, which is the solution space to the first two equations:
\[
  \mathbf{x}=
  \left(
    \begin{array}{c}
      0\\0\\c_1\\-c_2
    \end{array}
  \right)^T
  +
  \alpha\left(
    \begin{array}{c}
      1\\0\\-1\\1
    \end{array}
  \right)^T
  +
  \beta\left(
    \begin{array}{c}
      0\\1\\-1\\2
    \end{array}
  \right)^T:=\mathbf v_0 + \alpha \mathbf v_1 + \beta \mathbf v_2,\quad
  \alpha,\beta\in\RR.
\]
This means that, for each $\mathbf c$, $\cP_{\mathbf c}$ lives in an affine
space parallel to the the linear space $L:=\lin(\mathbf v_1,\mathbf
v_2)$. In Figure \ref{figure:1}, we show how $\cP_{\mathbf c}$ changes
as $\mathbf c$ takes values in the following set, which we chose because it was simple representation-wise:
\[
  (c_1,c_2)\in\{(1-t,t)|t\in[0,1]\}.
\]   

\begin{center}
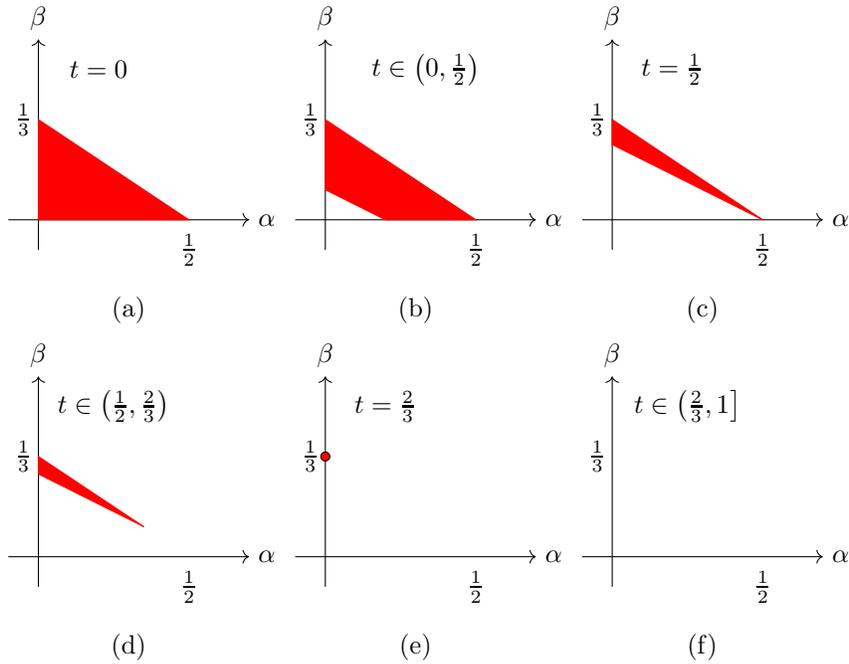
\begin{figure}
\begin{tikzpicture}[scale=3]
    \coordinate (P) at (0,0);
    \coordinate (Q) at (0,1/3);
    \coordinate (R) at (1/2,0);
    \draw[->] (-0.1, 0) -- (0.7, 0) node[right] {$\alpha$};
    \draw[->] (0, -0.1) -- (0, 0.6) node[above] {$\beta$};
    \draw[red, fill=red] (Q) -- (P) -- (R) -- (Q) -- cycle;
    \node at (1/2,-0.1) {$\frac{1}{2}$};
    \node at (-0.05,1/3) {$\frac{1}{3}$};
    \node at (1/5,1/2) {$t=0$};
    \node at (0.3,-0.3) {(a)};
\end{tikzpicture}
\begin{tikzpicture}[scale=3]
    \coordinate (M) at (0.2,0);
    \coordinate (P) at (0,0.1);
    \coordinate (Q) at (0,1/3);
    \coordinate (R) at (1/2,0);
    \draw[->] (-0.1, 0) -- (0.7, 0) node[right] {$\alpha$};
    \draw[->] (0, -0.1) -- (0, 0.6) node[above] {$\beta$};
    \draw[red, fill=red] (M) -- (P) -- (Q) -- (R) -- (M) -- cycle;
    \node at (1/2,-0.1) {$\frac{1}{2}$};
    \node at (-0.05,1/3) {$\frac{1}{3}$};
    \node at (1/3,1/2) {$t\in\left(0,\frac{1}{2}\right)$};
    \node at (0.3,-0.3) {(b)};
\end{tikzpicture}
\begin{tikzpicture}[scale=3]
    \coordinate (P) at (0,1/3);
    \coordinate (Q) at (0,1/4);
    \coordinate (R) at (1/2,0);
    \draw[->] (-0.1, 0) -- (0.7, 0) node[right] {$\alpha$};
    \draw[->] (0, -0.1) -- (0, 0.6) node[above] {$\beta$};
    \draw[red, fill=red] (Q) -- (P) -- (R) -- (Q) -- cycle;
    \node at (1/2,-0.1) {$\frac{1}{2}$};
    \node at (-0.05,1/3) {$\frac{1}{3}$};
    \node at (1/5,1/2) {$t=\frac{1}{2}$};
    \node at (0.3,-0.3) {(c)};
\end{tikzpicture}

\begin{tikzpicture}[scale=3]
    \coordinate (P) at (0.35,0.1);
    \coordinate (Q) at (0,0.55/2);
    \coordinate (R) at (0,1/3);
    \draw[->] (-0.1, 0) -- (0.7, 0) node[right] {$\alpha$};
    \draw[->] (0, -0.1) -- (0, 0.6) node[above] {$\beta$};
    \draw[red, fill=red] (Q) -- (P) -- (R) -- (Q) -- cycle;
    \node at (1/2,-0.1) {$\frac{1}{2}$};
    \node at (-0.05,1/3) {$\frac{1}{3}$};
    \node at (1/4,1/2) {$t\in\left(\frac{1}{2},\frac{2}{3}\right)$};
    \node at (0.3,-0.3) {(d)};
\end{tikzpicture}
\begin{tikzpicture}[scale=3]
    \draw[->] (-0.1, 0) -- (0.7, 0) node[right] {$\alpha$};
    \draw[->] (0, -0.1) -- (0, 0.6) node[above] {$\beta$};
    \node [draw, shape = circle, fill = red,
    minimum size = 0.12cm, inner sep=0pt] at (0,1/3){};
    \node at (1/2,-0.1) {$\frac{1}{2}$};
    \node at (-0.05,1/3) {$\frac{1}{3}$};
    \node at (0.2,1/2) {$t=\frac{2}{3}$};
    \node at (0.3,-0.3) {(e)};
\end{tikzpicture}
\begin{tikzpicture}[scale=3]
    \draw[->] (-0.1, 0) -- (0.7, 0) node[right] {$\alpha$};
    \draw[->] (0, -0.1) -- (0, 0.6) node[above] {$\beta$};
    \node at (1/2,-0.1) {$\frac{1}{2}$};
    \node at (-0.05,1/3) {$\frac{1}{3}$};
    \node at (1/4,1/2) {$t\in\left(\frac{2}{3},1\right]$};
    \node at (0.3,-0.3) {(f)};
\end{tikzpicture}
\caption{We chose $(c_1,c_2)$ to move on a segment
$(c_1,c_2)\in\{(1-t,t)|t\in[0,1]\}$. 
$\cP_{\mathbf c}$ changes with $t$: we start with a triangle (fig. a), which 
becomes a quadrilateral (fig. b), which then becomes a triangle (fig. c and 
d), which collapses to a point (fig. e); eventually, $\cP_{\mathbf c}$ becomes 
empty (fig. f).}
\label{figure:1}
\end{figure}
\end{center}
\end{example}

The relevance of invariant polyhedra lays in the fact that they are
related to siphons, which are combinatorial objects that indicate
which concentrations are allowed to vanish at steady states. Siphons have also 
proved to be useful in finding partial
solutions for the Global Atractor Conjecture \cite{CRACIUN20091551}.

\begin{defn}
Consider a mass-action network $\fN$ with variables  $\mathcal X$. A 
\emph{siphon} 
of $\fN$ is a nonempty subset $\mathcal S$ of $\mathcal X$ with the property 
that, given an arbitrary arrow $m\to m'$ of $\fN$, if
$\mathcal S \cap \mathcal M' \ne \emptyset$, then 
$\mathcal S \cap \mathcal M \ne \emptyset$ (cf., \cite{alg-030}).
\end{defn}

Note that, according to the previous definition, $\cX$ is a siphon of $\fN$. 
However, we often want to know whether a siphon contains smaller siphons 
itself. This leads us the the definition of minimal siphon.

\begin{defn}
A siphon is \emph{minimal} if none of its proper subsets is a 
siphon.
\end{defn}

\begin{example}
\label{ex:3}
In example \ref{ex:2} we have the following siphons:
$\{1,2,3,4\}$, $\{1,2,3\}$, $\{1,3,4\}$, $\{2,3,4\}$, $\{1,3\}$, $\{2,3\}$.
The only minimal ones are $\{1,3\}$ and $\{2,3\}$.
\end{example}

\begin{remark}
In Example \ref{ex:3} we can see that, in general, minimal siphons do not form 
a partition of $\cX$. This is for two reasons. The first one, is that minimal 
siphons are, in general, not disjoint. The second one, is that, in general, 
minimal siphons are not a cover for $\cX$.    
\end{remark}

Usually one decorates the vertices of the abstract shapes of
invariant polyhedra by the product of the variables whose supports coincide
with these vertices (see \cite{alg-030} for a method to compute all
combinatorial types of invariant polyhedra and their decorations). We
call such representation a \emph{decorated abstract invariant
polyhedron}.

\begin{defn}
A set of species $\Sigma = \{x_i\ | \ i\in I\subseteq[n]\}$ is called
an \emph{invariant polyhedral support} if all combinatorial
types of decorated abstract invariant polyhedra are invariant under
the map $x_i\mapsto x_j$ for all $x_i,x_j\in\Sigma$. An invariant polyhedral 
support is \emph{maximal} if it is maximal with respect to inclusion.
\end{defn}

\begin{remark}
In order to
compute the maximal invariant polyhedral supports, one needs to
compute the chamber decomposition in the space of linear conserved
quantities, and decorate each ray with the corresponding 
variables. The set of all variables decorating an individual ray is a
maximal invariant polyhedral support.
\end{remark}

A strategy towards solving equation~\eqref{eq:st-st} is by first
linearising the vector $\text{diag}(k)x^{Y_e}$, and then by imposing
binomial relations among its coordinates (cf.,
\cite{conradi2008multistationarity}). For the linearisation, note that
$x$ is nonnegative, so information about the \emph{stoichiometric
  cone} $\ker(Y_p-Y_e)\cap\RR^r_{\ge0}$ is enough. For the second
step, the authors of \cite{conradi2012multistationarity} introduced a
key concept, namely \emph{clusters}. \emph{Clusters} and special
partitions of them, \emph{preclusters}, were used in author's thesis
to show that a special class of mass-action networks, dynamical
systems with the isolation property, have toric steady states.

\begin{defn} Consider the graph $\mathfrak G$ over $[r]$
which has an edge $\{i,j\}$ whenever the $i^{\text{th}}$ and the
$j^{\text{th}}$ arrows have the same source and let $n_1,\ldots,n_r$ denote
the rows of a matrix with columns the rays of the stoichiometric
cone. If to the graph $\mathfrak G$ we adjoin the edges $\{i,s\}$ and
$\{j,s\}$ whenever $n_s$ is in the span of $\{n_i,n_j\}$ for each edge
$\{i,j\}$ of $\mathfrak G$, we obtain the \emph{preclustering
graph}. A \emph{precluster} is a connected component of the
preclustering graph.
\end{defn}

\section{The dual of a mass-action network}

There is a clear duality between the kernel of the stoichiometric matrix and 
the conservation space. This duality can be interpreted as a duality
between chemical species and rate constants. Consequently, we define
the dual of a mass-action network network to be
\[
\left(
\begin{tikzcd}
\bx^{Y_e}\arrow{r}{\bk}& \bx^{Y_p}
\end{tikzcd}\right)^* =
\begin{tikzcd}
\bk^{Y^T_e}\arrow{r}{\bx}& \bk^{Y^T_p}
\end{tikzcd}.
\]

For conservative systems, the duality between the conservation and
the stoichiometric space can be interpreted in terms of a duality between the 
corresponding nonnegative cones. We recall the following definition:

\begin{defn}
An \emph{internal cycle} of a mass-action network $\fN$ is a minimal multiset $C$ of 
the arrows of $\fN$ such that the monomial obtained from the product of all 
source monomials of arrows indexed by $C$ is equal to the monomial obtained 
from the product of all sink monomials of arrows indexed by $C$. The algebraic version of internal 
cycles are also called \emph{nonnegative elementary flux modes} in the literature.
\end{defn}

As a direct consequence, we get the following result:

\begin{thm}
\label{thm:duality}
For conservative mass-action networks, the set of conserved quantities is dual to the set of internal cycles.
\end{thm}
\begin{proof}
As there is a one to one correspondence between the extreme rays of
the stoichiometric cone and the minimal cycles of the mass-action
network (cf., 
\cite[Proposition~4.1]{schuster1994elementary}), one can conclude that 
conserved quantities are dual to internal cycles.
\end{proof}

This leads us to the following conjecture

\begin{conj}
\label{conj:duality}
For conservative mass-action networks that do not have two species
with exactly the same rates and have at least one positive steady
state, there is a duality relation between the sets of preclusters and
of maximal invariant polyhedral supports.
\end{conj}

\begin{remark}
  The condition $\dot x_i \neq \dot x_j$ for $i \neq j$ is necessary
  in order to make sure no further clustering on the left kernel is
  done. We require that there is at least one positive steady state so
  that the internal cycles are a refinement of preclusters.
\end{remark}

We conclude with the following question:

\begin{question}
Is there a duality relation between preclusters and siphons?
\end{question}

A (partial) positive answer to the previous question is supported by
the close relation between siphons and maximal invariant polyhedral
supports. On the other hand, it would be interesting to understand
whether a possible duality between preclusters and siphons would have
consequences for toric dynamical systems. It is well known that the
Global Attractor Conjecture holds in small dimensions
\cite{CRACIUN20091551}, and we wonder whether there is a similar
argument, through duality, in small codimensions.

\begin{example}
Consider the following network:
\[
\begin{tikzcd}[ampersand replacement=\&, row sep=small, every
arrow/.append style={shift left=.5}]
\displaystyle
\mathfrak {N}:\quad x_3x_4 \ar{r}{k_1} 
\& \displaystyle x_1x_2, 
\qquad \displaystyle x_1 \ar{r}{k_2} 
\& \displaystyle x_3 \end{tikzcd}
\]
Its dynamics is given by:    
\[
\left\{
\begin{array}{ll}
\dot x_1 = & \ \ \ \! k_1x_3x_4 - k_2x_1, \\
\dot x_2 = & \  \ \ \! k_1x_3x_4, \\
\dot x_3 = & -k_1x_3x_4 + k_2x_1, \\
\dot x_4 = & -k_1x_3x_4,\\
\end{array}
\right.
\qquad
x_i\ge0, \ i \in [4].
\]
Its conservation laws are
$x_1+x_3=c_1$ and
$x_2+x_4=c_2$. It has one combinatorial type of invariant polyhedron, which we decorate 
as:
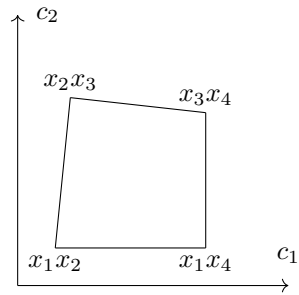
\[
\begin{tikzpicture}[scale=0.8]
\draw[->] (0,0)--(3.6,0);
\draw[->] (0,0)--(0,3.2);
\node at (3.6,0.4) {$c_1$};
\node at (0.4,3.2) {$c_2$};
\draw (0.5,0.5)--(0.7,2.5);
\draw (0.7,2.5)--(2.5,2.3);
\draw (2.5,2.3)--(2.5,0.5);
\draw (2.5,0.5)--(0.5,0.5);
\node at (0.5,0.3) {{$x_1x_2$}};
\node at (2.5,0.3) {{$x_1x_4$}};
\node at (0.7,2.7) {{$x_2x_3$}};
\node at (2.5,2.5) {{$x_3x_4$}};
\end{tikzpicture}
\]
     
Its minimal siphons are $\{\{x_4\},\{x_1,x_3\}\}$ and its MIPS are
$\{\{x_1,x_3\}$, $\{x_2,x_4\}\}$. The dual network is:
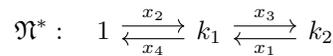
\[
\begin{tikzcd}[ampersand replacement=\&, row sep=small, every
arrow/.append style={shift left=1}]
\displaystyle \mathfrak {N}^*:\quad 1 \ar{r}{x_2}
\& \displaystyle k_1 \ar{l}{x_4}\ar{r}{x_3} \& k_2 \ar{l}{x_1}
\end{tikzcd}
\]
Finally, the preclusters of $\mathfrak N^*$ are $\{\{x_1,x_2,x_3,x_4\}\}$ and 
its internal cycles are $\{\{x_1,x_3\},\{x_2,x_4\}\}$. Note how the internal 
cycles partition of the only precluster.
\end{example}

We finish this section with the following scheme of the
state of art:
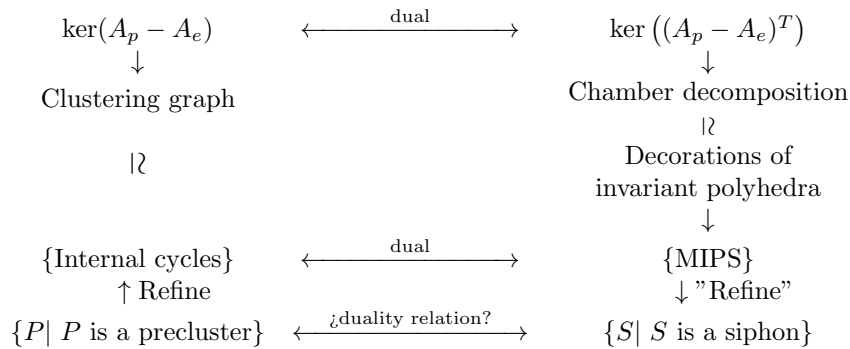
\[
\begin{array}{ccc}
\ker (Y_{p}-Y_{e})
&\xleftrightarrow{
\qquad \quad \text{dual}
\qquad \quad }&\ker\left((Y_{p}-Y_{e})^{T}\right)\\
\downarrow&&\downarrow\\
\begin{array}{c}
\text{Clustering graph}\\
\\
{\mbox{\Large $\simeqd$} } \\
\vspace{0.25cm}
\end{array}
&&
\begin{array}{c}
\text{Chamber decomposition}\\
\simeqd\\
\text{Decorations of}\\
\text{invariant polyhedra}\\
\downarrow\\
\end{array}\\
\{\text{Internal cycles}\}& \xleftrightarrow{
\qquad \quad \text{dual}\qquad \quad}&
\{\text{MIPS}\}\\
\qquad \uparrow{\text{Refine}}&&\qquad \downarrow{\text{"Refine"}}\\
\{P|\ P \text{ is a precluster}\}&\xleftrightarrow{
\quad \text{¿duality relation?}\quad}& \{S|\ S \text{ is a siphon}\}
\end{array}
\]

Finally, we would like to mention that the notion of duality
considered in this paper already appeared in 2020 in an earlier
preprint of the author \cite{iosifDualHal}. Only recently did the
author become aware that, around the same time, Blanchini and Giordano
\cite{blanchiniDual} introduced essentially the same notion of duality
in a broader dynamical context. These works were developed independently
and emphasise different aspects of this duality.

\section*{Acknowledgments}
I would like to thank Franco Blanchini and Giulia Giordano for reading
this article and commenting on it. This research started during the
author's thesis. The beginning of this work was funded by the projects
DFG 284057449 and DFG-RTG MathCore 314838170. The main part of this
work was funded by the Project 2025/SOLCON-160677 from Rey Juan Carlos
University. Special thanks to Lamprini Ananiadi for her valuable
perspectives and frequent dialogues on the subject.

\section*{Statement about the use of AI}
During the preparation of this work the author used Copilot and Gemini
in order to locate language errors and obvious typos.  Copilot was
also used to suggest a list of scientific journals where the present
article could better fit. In case AI was used to suggest other
scientific sources (e.g., book chapters, articles), all recommended
sources were verified and consulted directly, rather than relying on
the automated output. After using these tools, the author implemented,
manually reviewed and edited the content, serving as the sole editor,
as needed, and takes full responsibility for the content of the
published article.


\end{document}